\documentclass{scrartcl}
\newif\ifspringer
\newif\ifsubmit
\newif\ifarxiv

\arxivtrue

\ifspringer
\usepackage[utf8]{inputenc}
\fi
\ifarxiv
\usepackage[utf8]{inputenc}
\fi

\ifspringer
\smartqed  
\journalname{Journal of Mathematical Biology}
\else
\ifarxiv
\else
\usepackage{academicons}
\usepackage{xcolor}
\definecolor{orcidlogocol}{HTML}{A6CE39}
\fi
\usepackage[misc,geometry]{ifsym}
\usepackage{microtype}
\fi

\usepackage[unicode=true,
            bookmarks=true,bookmarksnumbered=false,bookmarksopen=false,
            hidelinks=True,
            colorlinks=false,breaklinks=false,backref=false]{hyperref}

\ifspringer
\else
\usepackage{amsmath}
\fi

\usepackage{amssymb}
\usepackage{mathrsfs} 
\usepackage{bm}
\usepackage{mathtools} 

\ifspringer
\else
\fi

\ifspringer
\newcommand\Wunderbar{\munderbar{W}}
\newcommand\Kunderbar{\munderbar{K}}
\newcommand\phiunderbar{\munderbar{\varphi}}
\newcommand\upOmega{\mathrm{\Omega}}
\newcommand\clOmega{\overline{\upOmega}}
\newcommand\upDelta{\mathrm{\Delta}}
\else
\newcommand\superunderline[3]{\mkern#1mu\underline{\mkern-#1mu#3\mkern-#2mu}\mkern#2mu}
\newcommand\Wunderbar{\superunderline{2.5}{8}{W}}
\newcommand\Kunderbar{\superunderline{3}{4.5}{K}}
\newcommand\phiunderbar{\superunderline{0}{3}{\varphi}}
\newcommand\clOmega{\overline{\Omega}}
\newcommand\upOmega{\Omega}
\newcommand\upDelta{\Delta}
\fi

\DeclareMathOperator{\diag}{diag}

\ifspringer
\else
\usepackage{amsthm}
\theoremstyle{plain}
\newtheorem{theorem}{Theorem}
\newtheorem{lemma}[theorem]{Lemma}
\theoremstyle{definition}
\newtheorem{remark}[theorem]{Remark}
\fi

\usepackage{enumitem}
\setenumerate[1]{label=(\roman*)} 

\usepackage{graphicx}

\ifspringer
\usepackage[caption=false]{subfig}
\else
\usepackage{caption}
\usepackage{subcaption}
\fi

\ifspringer
\usepackage{natbib}
\newcommand\wrapcite[1]{({#1})}
\else
\usepackage{doi}
\usepackage[style=numeric,giveninits=true,uniquename=init,backend=biber,maxcitenames=2,maxbibnames=20]{biblatex}
\AtEveryCite{%
}
\def\cite{\parencite}
\newcommand\wrapcite[1]{#1}
\fi

\ifspringer
\else
\fi

\newcommand\springerbreak{\ifspringer\linebreak\fi}

\ifspringer
\else
\title{Ecological invasion in competition-diffusion systems when the exotic species is either very strong or very weak}
\fi

\hypersetup{pdftitle={Ecological invasion in competition-diffusion systems when the exotic species is either very strong or very weak},
            pdfauthor={Lorenzo Contento, Danielle Hilhorst, Masayasu Mimura},
            pdfstartview=FitV}

\newcommand\MyAck{This research was performed in the context of the CNRS GDRI ReaDiNet. %
LC has been supported by the Meiji Institute for Advanced Study of Mathematical Sciences %
and by JSPS KAKENHI Grant-in-Aid for Research Activity Start-up No.~JP16H07254. %
MM has been partially supported by JSPS KAKENHI Grant No.~15K13462.}

\begin{document}

\ifspringer
\else
\author{Lorenzo Contento%
\thanks{Meiji Institute for Advanced Study of Mathematical Sciences, Meiji University, Tōkyō 164-8525, Japan \newline%
       \Letter\hspace{2mm}\texttt{lorenzo.contento@gmail.com}
       \ifarxiv
       \else
       \hspace{3mm}
       \href{https://orcid.org/0000-0002-7901-2172}{\textcolor{orcidlogocol}{\aiOrcid} \texttt{orcid.org/0000-0002-7901-2172}}
       \fi
}
\and Danielle Hilhorst
\thanks{CNRS and Laboratoire de Mathématiques, University of Paris-Sud, 91405 Orsay Cedex, France}
\and Masayasu Mimura
\thanks{Department of Mathematical Engineering, Musashino University, Tōkyō 135-8181, Japan}
}
\fi

\maketitle

\begin{abstract}
    Reaction-diffusion systems with a Lotka-Volterra-type reaction \springerbreak term,
    also known as competition-diffusion systems,
    have been used to investigate
    the dynamics of the competition
    among $m$ ecological species
    for a limited resource necessary to their survival and growth.
    Notwithstanding their rather simple mathematical structure,
    such systems may display quite interesting behaviours.
    In particular, while for $m=2$ no coexistence of the two species is usually possible,
    if $m \ge 3$ we may observe coexistence of all or a subset of the species,
    sensitively depending on the parameter values.
    Such coexistence can take the form of very complex spatio-temporal patterns and oscillations.

    Unfortunately,
    at the moment there are no known tools
    for a complete analytical study of such systems for $m \ge 3$.
    This means that establishing general criteria for the occurrence of coexistence appears to be very hard.
    In this paper we will instead give some criteria for the non-coexistence of species,
    motivated by the ecological problem of the invasion of an ecosystem by an exotic species.
    We will show that
    when the environment is very favourable to the invading species
    the invasion will always be successful
    and the native species will be driven to extinction.
    On the other hand, if the environment is not favourable enough,
    the invasion will always fail.
\ifspringer
\else
\medskip
\newline
\noindent
{\scshape Keywords}
\hspace{2mm}
competition-diffusion system \textperiodcentered\ifarxiv\ \fi ecological invasion \textperiodcentered\ifarxiv\ \fi competitive exclusion \textperiodcentered\ifarxiv\ \fi large-time behaviour \textperiodcentered\ifarxiv\ \fi singular limit \textperiodcentered\ifarxiv\ \fi comparison principle
\medskip
\newline
\noindent
{\scshape Mathematics Subject Classification}
\hspace{2mm}
35Q92 \textperiodcentered\ifarxiv\ \fi 92D25 \textperiodcentered\ifarxiv\ \fi 35K57 \textperiodcentered\ifarxiv\ \fi 35B25 \textperiodcentered\ifarxiv\ \fi 35B40 \textperiodcentered\ifarxiv\ \fi 35B51
\fi
\end{abstract}

\section{Introduction}
The understanding of the mechanisms
behind the rich biodiversity observed in nature
is a central problem in theoretical ecology.
It is a generally accepted fact that when two or more species
are competing for the same limited resources
in a constant and homogeneous environment
which is isolated from external influences,
they cannot coexist and all but one species will become extinct;
this is known as the \emph{competitive-exclusion principle}
and has been experimentally confirmed
for cultures of microorganisms \wrapcite{\cite{gause}}.
However, in real ecosystems a high number of coexisting species
is often observed also in places where resources are scarce.
A famous example of this apparent contradiction with the principle
is Hutchinson's paradox of the plankton \wrapcite{\cite{plankton}}:
a high number of phytoplankton species are able to coexist,
even if they all compete for the same resources.
Traditionally theoretical ecologists have explained this biodiversity
by observing that natural environments are inhomogeneous in space and/or time,
so that the principle does not apply.
Thus, even species which are competing for the same resource may coexist,
each being dominant in a particular zone or season,
without any equilibrium being reached.

Mathematical models for the competition between species
can aid in the understanding of this problem.
In the case where only two species are present,
it has been shown that a reaction-diffusion system
with Lotka-Volterra-like reaction terms
(from here on called a \emph{competition-diffusion system})
with constant parameters (i.e., modeling a homogeneous environment)
always displays competitive exclusion
if the space domain is convex \wrapcite{\cite{kishimoto,hirsch}}.
Non-convex domains may allow for stable coexistence equilibria
in which the species segregate spatially \wrapcite{\cite{matanomimura}},
but this can be interpreted ecologically as being due to immigration effects,
a violation of the hypotheses of the competitive-exclusion principle.
Another example of a mechanism which leads to coexistence
is the addition of cross-diffusion \wrapcite{\cite{shigesada}};
since this amounts to the species avoiding each other and nearly not competing,
it is again a failure of the principle's hypotheses.

It has been recently shown that, when three or more species are considered,
dynamical coexistence is possible even in convex homogeneous environments
with only random dispersal \wrapcite{\cite{morozov,tohmaPaper,contentoThesis,contento}}.
This is due to the effect of indirect competition between the species,
under the form of the so-called \emph{cyclic competition}.
The competition-diffusion system in this case has the form
\begin{equation*}
    \left\{
    \begin{alignedat}{6}
        u_t &= \, \upDelta u &&+ (r_1 &&- u &&- b_{12} v &&- b_{13} w&&) \, u, \\
        v_t &= \, \upDelta v &&+ (r_2 &&- v &&- b_{21} u &&- b_{23} w&&) \, v, \\
        w_t &= \, \upDelta w &&+ (r_3 &&- w &&- b_{31} u &&- b_{32} v&&) \, w.
    \end{alignedat}
    \right.
\end{equation*}
In particular,
in
\ifspringer
\else
\wrapcite{\cite{tohmaPaper,contentoThesis,contento}}
\fi
the following ecological situation is considered.
An ecosystem which is inhabited by two native species $u$ and $v$
which are usually unable to coexist
is invaded by a third, exotic species $w$ from outside.
The parameter $r_3$ represents the suitability of the new environment for the invader.
Intuitively, if $r_3$ is very small the invasion should fail,
while if $r_3$ is very large the two native species should be supplanted by $w$.
Then, coexistence is possible only for intermediate values of $r_3$.

This line of reasoning can be extended
to the general case in which we have $m$ different competing species.
Let us choose one species,
which without loss of generality can always be thought to be the $m$-th one.
If $r_m$, the intrinsic growth rate of the $m$-th species,
is very large,
then the $m$-th species will be able to invade
an ecosystem occupied by the first $m-1$ species,
completely replacing them
(see the numerical simulation in Figure~\ref{fig:rm_large}).
If on the other hand $r_m$ is very small,
the invasion will not succeed
and the $m$-th species will go extinct
(see Figure~\ref{fig:rm_small}).
Note that if $m > 3$ the remaining species may still be able to coexist.
Then, coexistence of all species is possible only for intermediate values of $r_m$,
when the invasion by the $m$-th species is successful
but its strength is not sufficient to drive the native species to extinction
(see Figure~\ref{fig:rm_intermediate}).
\begin{figure}
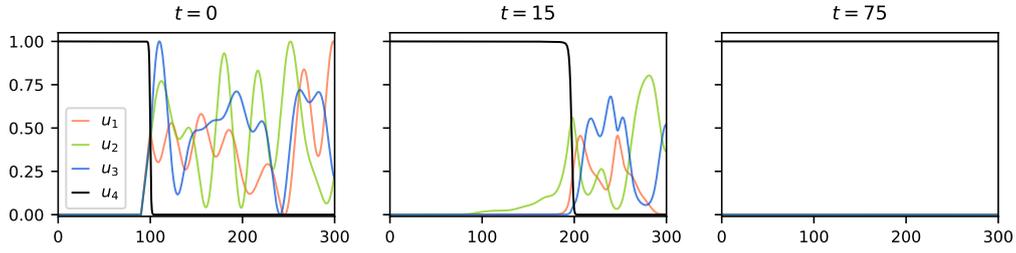
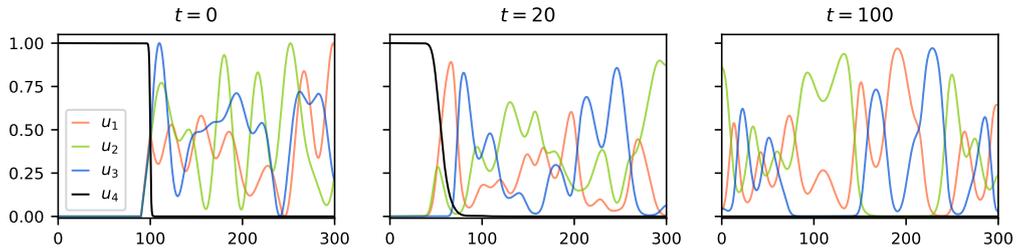
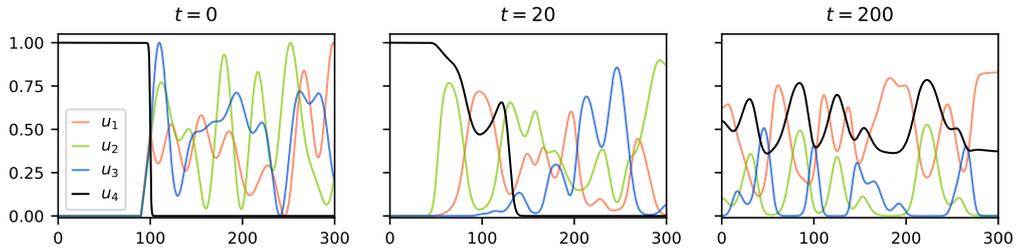

    \centering
    \ifspringer
\else
    \begin{subfigure}[b]{0.9\linewidth}
        \centering
        \ifarxiv
        \includegraphics[width=\textwidth]{{{figure1a}}}
        \else
        \includegraphics[width=\textwidth]{{{graphics/invasion_4_species_rm=7}}}
        \fi
        \caption{Simulation for $r_4=7$.
                 The invasion is successful
                 and the native species are driven extinct.}
        \label{fig:rm_large}
    \end{subfigure}

    \begin{subfigure}[b]{0.9\linewidth}
        \centering
        \ifarxiv
        \includegraphics[width=\textwidth]{{{figure1b}}}
        \else
        \includegraphics[width=\textwidth]{{{graphics/invasion_4_species_rm=0.3}}}
        \fi
        \caption{Simulation for $r_4=0.3$.
                 The invasion fails and $u_4$ goes extinct.
                 For this particular choice of parameters,
                 the remaining species are able to coexist
                 and no additional extinction occurs.}
        \label{fig:rm_small}
    \end{subfigure}

    \begin{subfigure}[b]{0.9\linewidth}
        \centering
        \ifarxiv
        \includegraphics[width=\textwidth]{{{figure1c}}}
        \else
        \includegraphics[width=\textwidth]{{{graphics/invasion_4_species_rm=1}}}
        \fi
        \caption{Simulation for $r_4=1$.
                 The invasion succeeds, but the native species are not completely replaced
                 and coexist with the invader.}
        \label{fig:rm_intermediate}
    \end{subfigure}
    \fi
    \caption{Numerical simulation of an ecosystem originally inhabited
             by three native species $u_i$, $i=1,2,3$,
             which is being invaded by a fourth exotic species $u_4$
             from the left side.
             Species densities are expressed
             as a fraction of the corresponding carrying capacity,
             i.e., the maximum density of each given species
             the environment can support in absence of its competitors.
             Each row represents a run with the same initial conditions
             but different values of $r_4$,
             the intrinsic growth rate of $u_4$.
             The other parameters are given by
             $d_1=1.90$, $d_2=1.85$, $d_3=0.89$, $d_4=1.70$,
             $r_1=r_2=r_3=1$,
             $b_{11}=b_{22}=b_{33}=b_{44}=1$,
             $b_{12}=0.70$, $b_{13}=1.67$, $b_{14}=0.44$,
             $b_{21}=1.73$, $b_{23}=0.81$, $b_{24}=0.16$,
             $b_{31}=0.26$, $b_{32}=1.81$, $b_{34}=0.99$,
             $b_{41}=0.78$, $b_{42}=0.10$, $b_{43}=0.94$.
             \ifarxiv
             \fi
             }
\end{figure}

In this paper we are mainly concerned
with studying mathematically (instead of just relying on numerical simulations)
the dependence of the system's behaviour on the parameter $r_m$.
We will only consider the extreme cases
in which such parameter is very large or very small.
The intermediate value case,
while very interesting since it allows for coexistence of all species,
is much more challenging to study analytically and will not be considered here.
In Section~\ref{sec:basic}
we recall the basic properties of the solutions
of the $m$-species competition-diffusion system.
In Section~\ref{sec:scalar}
we will study the scalar case $m=1$,
i.e., the Fisher-KPP equation,
and its limiting behaviour.
In Section~\ref{sec:rm_large}
we consider the case in which $r_m$ is large
and show that the first $m-1$ species become extinct.
We study this case first as a singular limit problem,
keeping the time instant fixed and letting $r_m$ go to infinity,
and then as a large-time problem,
choosing $r_m$ sufficiently large but finite
and studying the behaviour of the solutions as time goes to infinity.
In Section~\ref{sec:rm_small}
we study in the same way the case
where $r_m$ is small and the $m$-th species disappears.

\section{Basic properties of the $m$-species competition-diffusion system}
\label{sec:basic}
Let us consider
the following initial value problem {(P)\label{eq:cds}}
for the $m$-species competition-diffusion system:
\begin{gather}
    \partial_t u_i
    =
    d_i \, \upDelta u_i + \left( r_i - \sum_{j=1}^{m} b_{ij} u_j \right) \, u_i
    \quad
    \text{in } \upOmega\times\left(0,\infty\right),
    \quad
    \ifspringer
\else
    \text{for all }
    \fi
    i=1,\dots,m,
    \tag{CD} \label{eq:cds:pde}
    \\
    \partial_\nu \bm{u} = 0
    \quad \text{on }\partial\upOmega\times\left(0,\infty\right),
    \tag{BC} \label{eq:cds:bc}
    \\
    \bm{u}(x,0)=\bm{u}_{0}(x)
    \quad \text{for all}\ x\in\upOmega,
    \tag{IC} \label{eq:cds:ic}
\end{gather}
where
\begin{itemize}
\item $\upOmega$ is an open and bounded subset of $\mathbb{R}^{n}$ with smooth boundary $\partial\upOmega$;
\item $\bm{u} = (u_1,\dots,u_m) :\upOmega\times\left[0,\infty\right)\to\mathbb{R}^{m}$ is the vector-valued function representing the species densities;
\item $d_i > 0$, $i=1,\dots,m$,
      are the diffusion coefficients,
      representing the degree of motility of each species;
\item $r_{i} > 0$, $i=1,\dots,m$,
      are the intrinsic growth rates,
      representing the rate of growth of each species in absence of competition (both inter- and intra-specific);
\item $b_{ii} > 0$, $i=1,\dots,m$,
      are the intra-specific competition coefficients,
      which measure the strength of the competition between members of the same species;
\item $b_{ij} \ge 0$, $i,j=1,\dots,m$, $i \ne j$,
      are the inter-specific competition coefficients,
      which measure the strength of the competition between members of different species;
\item \eqref{eq:cds:bc} are zero-flux boundary conditions,
      where $\nu$ is the unit vector normal to $\partial\upOmega$,
      which model the fact that the ecosystem is closed and emigration/immigration is impossible;
\item $\bm{u}_{0}\in \mathcal{C}(\clOmega,\mathbb{R}^{m})$ is a given initial condition such that $\bm{u}_{0} \ge 0$.
\end{itemize}
\begin{remark}
    In this paper we are concerned with the case of competition in a homogeneous environment
    and thus we will suppose that all parameters are constant in space and time.
\end{remark}
\begin{remark}
    The value $r_{i}/b_{ii}$ is the carrying capacity of the $i$-th species,
    i.e., the maximum density of that species that the ecosystem can support
    in absence of inter-specific competition.
    In particular, if only one species is present,
    it will eventually reach its carrying capacity.
    This is a well known property of \hyperref[eq:cds]{(P)} for $m = 1$,
    i.e., the initial value problem for the Fisher-KPP equation,
    and will be reviewed in Section~\ref{sec:scalar}.
\end{remark}

Regarding the solutions of problem \hyperref[eq:cds]{(P)},
we have the following result.
\begin{theorem}
    \label{thm:existence}
    Problem \hyperref[eq:cds]{(P)}
    has a unique classical solution,
    namely a function $\bm{u}$ such that
    \begin{equation*}
        \bm{u} \in \mathcal{C}(\clOmega \times \left[0,\infty\right), \mathbb{R}^m)
                   \cap
                   \mathcal{C}^{2,1}(\clOmega \times \left(0,\infty\right), \mathbb{R}^m),
    \end{equation*}
    which satisfies \eqref{eq:cds:pde} pointwise,
    along with the boundary and initial conditions \eqref{eq:cds:bc} and \eqref{eq:cds:ic}, respectively.
    Moreover, for every $i=1,\dots,m$ we have
    \begin{equation}
        0 \le u_i
          \le \max\left\{ \frac{r_{i}}{b_{ii}},{\left\lVert u_{0,i}\right\rVert}_{\mathcal{C}(\clOmega)}\right\}
          \eqqcolon M_i
        \quad
        \text{in } \clOmega \times \left[ 0, \infty \right).
        \label{eq:boundedness}
    \end{equation}
    In particular,
    for every $i=1,\dots,m$
    such that $u_{0,i}$ is not identically equal to zero,
    we have that
    \begin{equation}
        u_i > 0
        \quad
        \text{in }  \clOmega \times \left( 0, \infty \right),
        \label{eq:solution_is_strictly_positive}
    \end{equation}
    and, for every $i=1,\dots,m$
    such that $u_{0,i} \le r_i / b_{ii}$
    and $u_{0,i}$ is not identically equal to $r_i / b_{ii}$,
    we have that
    \begin{equation*}
        u_i < \frac{r_i}{b_{ii}}
        \quad
        \text{in }  \clOmega \times \left( 0, \infty \right).
    \end{equation*}
\end{theorem}
\begin{proof}
    The existence of a local solution with the required regularity
    follows from \cite[Proposition 7.3.2]{lunardi}.
    In order to show global existence,
    it is then sufficient to prove that \eqref{eq:boundedness} holds.
    First, for every $i=1,\dots,m$,
    we apply the comparison principle
    to the $i$-th equation in \eqref{eq:cds:pde}.
    By comparing $u_i$ with the constant function $0$
    we can show that $u_i \ge 0$ on the domain of definition.
    Then, by using the non-negativity of all $u_i$, $i=1,\dots,m$,
    we can apply the comparison principle again
    for $u_i$ and the constant function $M_i$,
    completing the proof of \eqref{eq:boundedness}.
    The last part of the theorem follows from the strong comparison principle.
\end{proof}

\section{The scalar case}
\label{sec:scalar}
We start by studying the scalar case,
in which \hyperref[eq:cds]{(P)} reduces to the scalar Fisher-KPP equation.
We first recall the explicit solution for the case when diffusion is absent,
i.e., for a logistic-type ordinary differential equation.
Using this formula we give uniform in space bounds from above and below
for solutions of the Fisher-KPP equation with constant coefficients.
Finally, we introduce a non-constant term $K$
to the expression for the growth rate
and give bounds for the resulting scalar reaction-diffusion equation.
The results obtained for this last form of the Fisher-KPP equation
will then be applicable directly to each scalar equation in \hyperref[eq:cds]{(P)}
by taking $K$ to be equal to the total effect of inter-specific competition.

The logistic ordinary differential equation
models the growth of a population
in presence of intra-specific competition.
If the initial population is non-zero,
then the population density will tend monotonically to the carrying capacity.
For reasons that will become apparent later,
we also need to treat the case in which the growth rate is negative,
which results in the population going extinct.
These results are formalized in the lemma below,
along with singular limit results for the growth rate going to $\pm\infty$.
\begin{lemma}[Explicit form of logistic growth]
    \label{lem:logistic_growth}
    Let $W^{(r,b)}(t,W_0)$
    be the solution of the following initial value problem
    for the logistic-type ordinary differential equation
    \begin{equation}
      \left\{
      \begin{aligned}
        & W_t  = \left( r - b \, W \right) W, \quad \text{in }\left(0,\infty\right),
        \\
        & W(0) = W_0,
      \end{aligned}
      \right.
      \label{eq:logistic_growth_ode}
    \end{equation}
    where $b > 0$ and $W_0 \ge 0$.
    If $W_0 = 0$, then $W \equiv 0$.
    If $W_0 = r/b$, then $W \equiv r/b$.
    If $r = 0$, then
    \begin{equation*}
        W^{(0,b)}(t,W_0) = \frac{W_0}{b\,W_0\,t + 1}.
    \end{equation*}
    If $r \neq 0$ instead (either $r>0$ or $r<0$), we have
    \begin{equation*}
        W^{(r,b)}(t,W_0) = \frac{r/b}{ 1 + \left( \frac{r/b}{W_0} - 1 \right) e^{- r \, t} }. 
    \end{equation*}

    Moreover, if $r > 0$ we have that
    \begin{equation}
        \lim_{t\to\infty} W^{(r,b)}(t,W_0) = \frac{r}{b},
        \label{eq:logistic_growth_ode:limit}
    \end{equation}
    where the convergence is monotone.
    In particular, $W^{(r,b)}(t,W_0)$ is increasing in $t$ if $W_0 < r/b$
    and decreasing in $t$ if $W_0 > r/b$.
    If $r \le 0$ instead,
    then
    \begin{equation}
        \lim_{t\to\infty} W^{(r,b)}(t,W_0) = 0,
        \label{eq:logistic_growth_ode:limit_r_negative}
    \end{equation}
    where the convergence is monotone decreasing.

    Finally, for every $\delta > 0$ we have
    \begin{gather}
        \lim_{r\to-\infty} {\left \lVert W^{(r,b)}(t,W_0) \right \rVert}_{\mathcal{C}([\delta,\infty))} = 0,
        \label{eq:logistic_growth_ode:limit_for_r_to_minus_infty}
        \\
        \lim_{r\to\infty} {\left \lVert \frac{W^{(r,b)}(t,W_0)}{r/b} - 1 \right \rVert}_{\mathcal{C}([\delta,\infty))} = 0.
        \label{eq:logistic_growth_ode:limit_for_r_to_infty:scaled}
    \end{gather}
\end{lemma}

We now add diffusion to problem \eqref{eq:logistic_growth_ode},
considering the problem of the logistic growth of a population
which moves randomly in space.
The resulting partial differential equation is the well-known Fisher-KPP equation
and its limiting behaviour is essentially unchanged from \eqref{eq:logistic_growth_ode}:
if the growth rate is positive and the initial data non-zero,
then solutions will tend to the carrying capacity $r/b$;
if instead the growth rate is negative, the population will go extinct.
This is a well known property but nonetheless we give a full proof here,
since we believe that it is a good introduction to the $m$-species case that will be studied in the following sections.
First, we prove that the higher the growth rate the faster the population will grow
and then we give lower and upper bounds to constrain the solutions,
which allows us to finally prove the convergence to the limit solution.

We will denote by $w^{(r,b,d,0)}(x,t,w_0)$ the unique classical solution
of the initial value problem {(F-KPP)\label{eq:fkpp}}
for \eqref{eq:cds:pde} in the case $m=1$,
that is, the Fisher-KPP equation
\begin{equation}
    w_t = d \, \upDelta w + ( r - b \, w ) \, w
    \quad
    \text{in } \upOmega \times \left( 0, \infty \right),
    \label{eq:convergence_to_cc:fisherKPP}
\end{equation}
with zero-flux boundary conditions on $\partial\upOmega$
and non-negative initial conditions $w_0 \in \mathcal{C}(\clOmega)$.
The reason behind the final $0$ superscript will become clear later.
In the Fisher-KPP equation \eqref{eq:convergence_to_cc:fisherKPP}
we suppose that $d \ge 0$ and $b > 0$,
but as in the logistic equation \eqref{eq:logistic_growth_ode}
we allow $r$ to be arbitrary, possibly negative.
Even in the case of negative growth rates,
it is easy to check that the results of Theorem~\ref{thm:existence} (case $m=1$) still hold,
so that $w^{(r,b,d,0)}(x,t,w_0)$ is well-defined for any choice of $r\in\mathbb{R}$.

\begin{lemma}
    \label{lem:fisher-kpp_is_increasing}
    Let $w^{(r,b,d,0)}(x,t,w_0)$ be the solution of problem~\hyperref[eq:fkpp]{(F-KPP)}
    with initial data $w_0$ non-negative.
    Then, $w^{(r,b,d,0)}(x,t,w_0)$ is non-decreasing in $r$.
\end{lemma}
\begin{proof}
We will apply the comparison principle.
Let $r' > r''$ arbitrary.
The functions $w^{(r',b,d,0)}(x,t,w_0)$ and $w^{(r'',b,d,0)}(x,t,w_0)$
solve \eqref{eq:convergence_to_cc:fisherKPP}
with $r=r'$ and $r=r''$ respectively.
Moreover, since such solutions are non-negative by \eqref{eq:boundedness},
$w^{(r',b,d,0)}(x,t,w_0)$
is a upper solution of problem~\hyperref[eq:fkpp]{(F-KPP)} with $r=r''$.
Then, $w^{(r',b,d,0)}(x,t,w_0) \ge w^{(r'',b,d,0)}(x,t,w_0)$
for every $(x,t) \in \clOmega \times \left[0,\infty\right)$.
\end{proof}

\begin{theorem}
    \label{thm:limits_for_fisherKPP}
    Let $w^{(r,b,d,0)}(x,t,w_0)$ be the classical solution of problem~\hyperref[eq:fkpp]{(F-KPP)}
    with $w_0$ non-negative and not identically equal to zero.
    Then,
    \begin{enumerate}
        \item if $r > 0$,
              the function $w^{(r,b,d,0)}(x,t,w_0)$
              converges to $r/b$
              uniformly on $\clOmega$ as $t \to \infty$;
        \item if $r \le 0$,
              the function $w^{(r,b,d,0)}(x,t,w_0)$
              converges to $0$
              uniformly on $\clOmega$ as $t \to \infty$;
        \item for any arbitrary $\delta > 0$, we have
              \begin{align}
                  \lim_{r \to -\infty} {\left \lVert w^{(r,b,d,0)}(x,t,w_0) \right \rVert}_{\mathcal{C}(\clOmega\times\left[\delta,\infty\right))} &= 0,
                  \label{eq:limits_for_fisherKPP:r_to_minus_infty}
                  \\
                  \lim_{r \to \infty} {\left \lVert \frac{w^{(r,b,d,0)}(x,t,w_0)}{r/b} - 1 \right \rVert}_{\mathcal{C}(\clOmega\times\left[\delta,\infty\right))} &= 0.
                  \label{eq:limits_for_fisherKPP:r_to_infty}
              \end{align}
    \end{enumerate}
\end{theorem}
\begin{proof}
    The proof consists in finding a couple of uniform in space bounds
    for the solution of problem~\hyperref[eq:fkpp]{(F-KPP)}
    by using the solutions to the logistic-type ordinary differential equation \eqref{eq:logistic_growth_ode}.
    Then, the results will follow from the convergence properties stated in Lemma~\ref{lem:logistic_growth}.

    Fix $\delta > 0$ and choose $\tau \in \left(0,\delta\right)$.
    Let
    \begin{align*}
        \bar{W}_0             &= \max_{x\in\clOmega} w_0(x),
        \\
        \Wunderbar_0^{(r)} &= \min_{x\in\clOmega} w^{(r,b,d,0)}(x,\tau,w_0).
    \end{align*}
    We have used the superscript ${}^{(r)}$ to highlight the dependence on the growth rate $r$.
    Since we supposed that $w_0$ is not identically equal to zero,
    we have that $\bar{W}_0 > 0$.
    Moreover, by \eqref{eq:solution_is_strictly_positive}
    we also get $\Wunderbar_0^{(r)} > 0$ for every $r\in\mathbb{R}$ and every choice of $\tau > 0$.
    We remark that this is not necessarily the case when $\tau=0$.

    Then, by the comparison principle,
    for every $(x,t) \in \clOmega \times \left[\tau,\infty\right)$
    we have that
    \begin{equation}
        W^{(r,b)}(t-\tau,\Wunderbar_0^{(r)}) \le w^{(r,b,d,0)}(x,t,w_0) \le W^{(r,b)}(t,\bar{W}_0),
        \label{eq:bounds_for_fisher-kpp}
    \end{equation}
    where $W^{(r,b)}(t,W_0)$ is defined as in Lemma~\ref{lem:logistic_growth}.
    Then, points (i) and (ii) immediately follow from
    \eqref{eq:logistic_growth_ode:limit}
    and
    \eqref{eq:logistic_growth_ode:limit_r_negative},
    respectively.
    Assertion \eqref{eq:limits_for_fisherKPP:r_to_minus_infty}
    is an immediate consequence of \eqref{eq:bounds_for_fisher-kpp}
    and \eqref{eq:logistic_growth_ode:limit_for_r_to_minus_infty}.

    We will now prove \eqref{eq:limits_for_fisherKPP:r_to_infty}.
    Thanks to \eqref{eq:bounds_for_fisher-kpp},
    it is sufficient to prove it separately for the lower and upper bound.
    For the lower bound some extra care is needed
    since the initial value $\Wunderbar_0^{(r)}$ depends on $r$.
    However, this difficulty can be easily overcome by observing that
    $\Wunderbar_0^{(r)}$ is increasing in $r$ by Lemma~\ref{lem:fisher-kpp_is_increasing}.
    In particular, Lemma~\ref{lem:fisher-kpp_is_increasing} implies that
    \begin{equation*}
        w^{(1,b,d,0)}(x,\tau,w_0) \le w^{(r,b,d,0)}(x,\tau,w_0),
        \quad
        \text{for all } x\in\clOmega,
        \text{ for all }r > 1,
    \end{equation*}
    so that $\Wunderbar_0^{(1)} \le \Wunderbar_0^{(r)}$
    for all $r >1$.

    Then, by the comparison principle applied to the initial value problem
    for the logistic-type differential equation \eqref{eq:logistic_growth_ode},
    we have that
    \begin{equation*}
        W^{(r,b)}(t-\tau,\Wunderbar_0^{(1)}) \le W^{(r,b)}(t-\tau,\Wunderbar_0^{(r)}),
        \quad
        \text{for all } t \ge 0,
        \text{ for all } r \ge 1.
    \end{equation*}
    This means that for $r$ large enough
    the inequalities \eqref{eq:bounds_for_fisher-kpp}
    imply that
    \begin{equation*}
        W^{(r,b)}(t-\tau,\Wunderbar_0^{(1)}) \le w^{(r,b,d,0)}(x,t,w_0) \le W^{(r,b)}(t,\bar{W}_0).
    \end{equation*}
    Since $\delta > \tau$ and the initial value $\Wunderbar_0^{(1)}$ does not depend on $r$,
    the left hand side divided by $r/b$ converges to $1$
    uniformly on $\clOmega \times \left[\delta,\infty\right)$
    by \eqref{eq:logistic_growth_ode:limit_for_r_to_infty:scaled}.
    The same can be said for the right hand side and so the proof of \eqref{eq:limits_for_fisherKPP:r_to_infty} is complete.
\end{proof}

We will now subtract a time and space dependent term $K(x,t) \ge 0$
from the growth rate of the Fisher-KPP equation.
Such term may for example represent the competitive interaction with other species,
which results in a lowered growth rate.
The equation which we obtain will be a general form of the equations in \eqref{eq:cds:pde}
and will be useful in studying the full system in the next sections.
Due to the non-constant effect of the term $K$,
the solution will not always converge as $t\to\infty$ as is the case for problem~\hyperref[eq:fkpp]{(F-KPP)},
only doing so when the growth rate is small enough.
In the other cases, we will only be able to show
that in the long run the solution
stays between two time-independent bounds.

We will denote by $w^{(r,b,d,K)}(x,t,w_0)$ the classical solution
(if it exists)
of the initial value problem {(F-KPP-K)\label{eq:fkpp_with_K}}
for the following reaction-diffusion equation
\begin{equation}
    w_t = d \, \upDelta w + ( r - K - b \, w ) \, w
    \quad
    \text{in } \upOmega \times \left( 0, \infty \right),
    \label{eq:fisherKPP_with_K}
\end{equation}
with zero-flux boundary conditions on $\partial\upOmega$
and non-negative initial conditions $w_0 \in \mathcal{C}(\clOmega)$.
In \eqref{eq:fisherKPP_with_K}
we suppose that $d \ge 0$, $r > 0$ and $b > 0$
and that the bounded function $K=K(x,t)$ is non-negative.
In the case $K \equiv 0$,
equation \eqref{eq:fisherKPP_with_K}
reduces to the standard Fisher-KPP equation \eqref{eq:convergence_to_cc:fisherKPP}.

\begin{remark}
    For the sake of brevity,
    we do not discuss under which conditions on $K$
    the classical solution of problem~\hyperref[eq:fkpp_with_K]{(F-KPP-K)} exists.
    This will not be a problem,
    since we will always apply the results of this section
    to single components of the vector solution $\bm{u}$ of problem~\hyperref[eq:cds]{(P)},
    whose existence is guaranteed by Theorem~\ref{thm:existence}.

    However, in order to correctly introduce the symbol $w^{(r,b,d,K)}(x,t,w_0)$,
    we need at least to know that the classical solution of problem~\hyperref[eq:fkpp_with_K]{(F-KPP-K)}
    is unique if it exists.
    This can be easily proven in a completely standard way.
    Let $w_1$ and $w_2$ be two classical solutions
    of problem~\hyperref[eq:fkpp_with_K]{(F-KPP-K)}
    for the same initial data $w_0$.
    We will show that they must coincide, thus proving the uniqueness of the classical solution.
    By substituting $w = w_1$ and $w = w_2$
    in \eqref{eq:fisherKPP_with_K}
    and subtracting the resulting equations,
    we obtain
    \begin{equation*}
        \varphi_t = d \, \upDelta \varphi + \left(r-K\right)\varphi - b \left(w_1 + w_2\right) \varphi,
    \end{equation*}
    where $\varphi = w_1 - w_2$.
    By multiplication of each side by $\varphi$ we get that
    \begin{equation*}
        \varphi \, \varphi_t
        =
        d \, \varphi \, \upDelta \varphi + \left(r-K\right)\varphi^2 - b \left(w_1 + w_2\right) \varphi^2
        \le
        d \, \varphi \, \upDelta \varphi + r \, \varphi^2,
    \end{equation*}
    where the inequality holds since $w_1, w_2 \ge 0$ by the comparison principle
    and $K \ge 0$ by hypothesis.
    By integration on $\upOmega$ and Green's first identity, we have
    \begin{align*}
        \frac{1}{2} \frac{\mathrm{d}}{\mathrm{d}t} {\left\lVert \varphi(t) \right\rVert}^{2}_{L^2(\upOmega)}
        &=
        \int_{\mathrlap{\upOmega}} \,\, \varphi \, \varphi_t
        \le
        d {\int_{\mathrlap{\upOmega}} \,\, \varphi \, \upDelta \varphi} + r {\int_{\mathrlap{\upOmega}} \,\, \varphi^2}
        \\
        &=
        -d {\int_{\mathrlap{\upOmega}} \, {\left( \nabla \varphi \right)}^2} + r {\int_{\mathrlap{\upOmega}} \,\, \varphi^2}
        \le r {\left\lVert \varphi(t) \right\rVert}^{2}_{L^2(\upOmega)}.
    \end{align*}
    Then, the application of the Gronwall lemma yields that
    \begin{equation*}
        {\left\lVert \varphi(t) \right\rVert}_{L^2(\upOmega)}
        \le
        e^{\frac{1}{2}rt}
        {\left\lVert \varphi(0) \right\rVert}_{L^2(\upOmega)}.
    \end{equation*}
    Since the initial value is the same for both $w_1$ and $w_2$,
    we have that $\varphi(0) = w_1(0) - w_2(0) \equiv 0$
    from which we conclude that $\varphi$ is identically zero at all times,
    which means that $w_1$ and $w_2$ coincide.
\end{remark}

\begin{theorem}
    \label{thm:limits_for_fisher-kpp_with_K}
    Let $w^{(r,b,d,K)}(x,t,w_0)$ be the classical solution (if it exists)
    of problem~\hyperref[eq:fkpp_with_K]{(F-KPP-K)}
    with $w_0$ non-negative and not identically equal to zero.
    Let $\delta > 0$ arbitrary
    and let $\Kunderbar$ and $\bar{K}$ be such that
    $0 \le \Kunderbar \le K(x,t) \le \bar{K}$
    for every $(x,t) \in \clOmega \times \left[\delta,\infty\right)$.
    Then, for every $(x,t) \in \clOmega \times \left[\delta,\infty\right)$
    we have that
    \begin{equation}
        w^{(r-\bar{K},b,d,0)}(x,t,w_0)
        \le
        w^{(r,b,d,K)}(x,t,w_0)
        \le
        w^{(r-\Kunderbar,b,d,0)}(x,t,w_0),\
        \label{eq:limits_for_fisher-kpp_with_K:bounds}
    \end{equation}
    Moreover, we have that
    \begin{enumerate}
        \item if $r \le \Kunderbar$,
              then $w^{(r,b,d,K)}(x,t,w_0)$ converges to $0$
              uniformly on $\clOmega$
              as $t \to \infty$;
        \item if $\Kunderbar < r \le \bar{K}$,
              then
              \begin{align*}
                  0
                  &\le
                  \adjustlimits
                  \liminf_{t \to \infty}
                      \min_{x\in\clOmega}
                         w^{(r,b,d,K)}(x,t,w_0)
                  \\
                  &\le
                  \adjustlimits
                  \limsup_{t \to \infty}
                      \max_{x\in\clOmega}
                         w^{(r,b,d,K)}(x,t,w_0)
                  \le
                  \frac{r-\Kunderbar}{b};
              \end{align*}
        \item if $r > \bar{K}$,
              then
              \begin{align*}
                  \frac{r-\bar{K}}{b}
                  &\le
                  \adjustlimits
                  \liminf_{t \to \infty}
                      \min_{x\in\clOmega}
                         w^{(r,b,d,K)}(x,t,w_0)
                  \\
                  &\le
                  \adjustlimits
                  \limsup_{t \to \infty}
                      \max_{x\in\clOmega}
                         w^{(r,b,d,K)}(x,t,w_0)
                  \le
                  \frac{r-\Kunderbar}{b};
              \end{align*}
        \item if the upper bound $\bar{K}$ can be chosen independently of $r$ for $r$ sufficiently large
              and if the solution $w^{(r,b,d,K)}(x,t,w_0)$ exists for all $r$ sufficiently large,
              then it holds that
              \begin{equation*}
                  \lim_{r\to\infty} {\left \lVert \frac{w^{(r,b,d,K)}(x,t,w_0)}{r/b} - 1 \right \rVert}_{\mathcal{C}(\clOmega\times\left[\delta,\infty\right))} = 0.
              \end{equation*}
    \end{enumerate}
\end{theorem}
\begin{proof}
    We start by proving \eqref{eq:limits_for_fisher-kpp_with_K:bounds}.
    First, observe that $w^{(r,b,d,K)}(x,t,w_0) \ge 0$
    by the comparison principle,
    since $w_0 \ge 0$
    and the constant function $0$
    is a solution of \eqref{eq:fisherKPP_with_K}.
    As a consequence, by using the fact that $r-\bar{K} \le r \le r-\Kunderbar$ by hypothesis,
    we obtain that
    \begin{align*}
        d \, \upDelta w + \left( r -\bar{K} - b \, w \right) w
        &\le
        w_t
        = d \, \upDelta w + \left( r - K - b \, w \right) w,
        \\
        d \, \upDelta w + \left( r - K - b \, w \right) w
        &=
        w_t
        \le
        d \, \upDelta w + \left( r -\Kunderbar - b \, w \right) w,
    \end{align*}
    where we have let $w=w^{(r,b,d,K)}(x,t,w_0)$ for the sake of brevity.
    This means that $w^{(r,b,d,K)}(x,t,w_0)$
    is respectively a upper solution of problem~\hyperref[eq:fkpp]{(F-KPP)} with growth rate equal to $r-\bar{K}$
    and a lower solution of problem~\hyperref[eq:fkpp]{(F-KPP)} with growth rate equal to $r-\Kunderbar$,
    so that we have \eqref{eq:limits_for_fisher-kpp_with_K:bounds} by the comparison principle.

    From \eqref{eq:limits_for_fisher-kpp_with_K:bounds},
    by taking the minimum over $\clOmega$ for the lower bound
    and the maximum over $\clOmega$ for the upper bound,
    we get
    \begin{align*}
        \min_{x\in\clOmega}
            w^{(r-\bar{K},b,d,0)}(x,t,w_0)
        &\le
        \min_{x\in\clOmega}
            w^{(r,b,d,K)}(x,t,w_0)
        \\
        &\le
        \max_{x\in\clOmega}
            w^{(r,b,d,K)}(x,t,w_0)
        \\
        &\le
        \max_{x\in\clOmega}
            w^{(r-\Kunderbar,b,d,0)}(x,t,w_0).
    \end{align*}
    We can take the liminf on both sides of the leftmost inequality
    and the limsup on both sides of the rightmost inequality,
    obtaining that
    \begin{align*}
        \adjustlimits
        \lim_{t\to\infty}
            \min_{x\in\clOmega}
                w^{(r-\bar{K},b,d,0)}(x,t,w_0)
        &\le
        \adjustlimits
        \liminf_{t \to \infty}
            \min_{x\in\clOmega}
               w^{(r,b,d,K)}(x,t,w_0)
        \nonumber
        \\
        &\le
        \adjustlimits
        \limsup_{t \to \infty}
            \max_{x\in\clOmega}
               w^{(r,b,d,K)}(x,t,w_0)
        \\
        &\le
        \adjustlimits
        \lim_{t\to\infty}
            \max_{x\in\clOmega}
                w^{(r-\Kunderbar,b,d,0)}(x,t,w_0).
    \end{align*}
    We remark that in the case of the lower and upper bounds,
    which are solutions of the Fisher-KPP equation \eqref{eq:convergence_to_cc:fisherKPP},
    the liminf can be reduced to a regular limit thanks to Theorem~\ref{thm:limits_for_fisherKPP}.
    In the case of the solutions to problem~\hyperref[eq:fkpp_with_K]{(F-KPP-K)},
    the limit may in general not exist.
    In particular, Theorem~\ref{thm:limits_for_fisherKPP} allows us to write explicitly the limit values,
    yielding
    \begin{align}
        \max\left\{ \frac{r-\bar{K}}{b}, 0 \right\}
        &\le
        \adjustlimits
        \liminf_{t \to \infty}
            \min_{x\in\clOmega}
               w^{(r,b,d,K)}(x,t,w_0)
        \nonumber
        \\
        &\le
        \adjustlimits
        \limsup_{t \to \infty}
            \max_{x\in\clOmega}
               w^{(r,b,d,K)}(x,t,w_0)
        \le
        \max\left\{ \frac{r-\Kunderbar}{b}, 0 \right\}.
        \label{eq:limits_for_fisher-kpp_with_K:limits}
    \end{align}

    Then, points (i), (ii) and (iii) are just specialization of \eqref{eq:limits_for_fisher-kpp_with_K:limits}
    for different values of $r$.
    Finally, point (iv) follows from point (iii),
    since we have that
    \begin{equation*}
        \lim_{r\to\infty} \frac{(r-\bar{K})/b}{r/b} = 1
        \quad
        \text{and}
        \quad
        \frac{r-\Kunderbar}{b} \le \frac{r}{b},
    \end{equation*}
    where the limit holds thanks to $\bar{K}$ being independent of $r$.
\end{proof}

\section{Behaviour when $r_m$ is large}
\label{sec:rm_large}
In this section we study the behaviour of the solutions of \hyperref[eq:cds]{(P)},
the $m$-species competition-diffusion system,
when one of the species
has a very large intrinsic growth rate.
Without any loss of generality we can suppose this species to be the $m$-th one,
having growth rate $r_m$.
From ecological considerations we expect all other species to become extinct,
while the density of the $m$-th one, left alone, reaches its carrying capacity.
In Theorem~\ref{thm:rm_to_infty}
we show that
all species but the $m$-th one go to zero uniformly in space and time
if we let $r_m$ tend to infinity.
Then, in Theorem~\ref{thm:rm_large}
we show that extinction of the first $m-1$ species
also occurs as $t$ tends to infinity
if $r_m$ is large enough (but still finite).

In this and the following sections we will denote by
$K_i$ the total effect on the growth rate of the $i$-th species
by the inter-specific competition with the other species.
Namely,
\begin{equation*}
    K_i(x,t) = \sum_{\substack{j=1\\j \ne i}}^{m} b_{ij} u_j(x,t),
\end{equation*}
so that \eqref{eq:cds:pde} becomes
\begin{equation}
\partial_t u_i
=
d_i \, \upDelta u_i + ( r_i - K_i - b_{ii} u_i ) \, u_i
\quad \text{in } \upOmega\times\left(0,\infty\right),
\quad \text{for all } i=1,\dots,m.
\label{eq:cds:pde:K_form}
\end{equation}
This form allows us to see more clearly that we can apply Theorem~\ref{thm:limits_for_fisher-kpp_with_K}
to each equation in \eqref{eq:cds:pde:K_form}.
In particular, using the notation introduced in the previous section, we have
\begin{equation}
    u_i(x,t) = w^{(r_i,b_{ii},d_i,K_i)}(x,t,u_{0,i}),
    \qquad
    \text{for all } i=1,\dots,m.
    \label{eq:cds:pde:K_form:single}
\end{equation}

Now we show that in the limit for $r_m \to \infty$ the $m$-th species will converge uniformly to its carrying capacity,
while the others will converge to $0$.
We remark that in order to prove the extinction of the $i$-th species, for $i=1,\dots,m-1$,
it is necessary to suppose that $b_{im} > 0$.
If this were not the case,
the $i$-th species would not be affected directly by the $m$-th one.
It would still be affected indirectly through the other species,
but, since no species can be benefited by a stronger invader
and since all inter-species interactions are competitive,
the $i$-th species is never penalized by higher densities of the $m$-th one.

\begin{theorem}
    \label{thm:rm_to_infty}
    Let $\bm{u}$
    be the classical solution of \hyperref[eq:cds]{(P)}
    with non-negative initial data $\bm{u}_0 \in \mathcal{C}(\clOmega,\mathbb{R}^m)$
    such $u_{0,m}$ is not identically zero.
    Suppose that
    $b_{im} > 0$
    for every $i=1,\dots,m-1$.
    Then, for every $\delta > 0$
    we have that
    \begin{gather}
        \lim_{r_m\to\infty}
            {\left \lVert u_i \right \rVert}_{\mathcal{C}(\clOmega\times\left[\delta,\infty\right))}
        = 0
        \qquad
        \text{for all }i=1,\dots,m-1,
        \label{eq:rm_to_infty:thesis_for_v}
        \\
        \lim_{r_m\to\infty}
            {\left \lVert \frac{u_m}{r_m/b_{mm}} - 1 \right \rVert}_{\mathcal{C}(\clOmega\times\left[\delta,\infty\right))}
        = 0.
        \label{eq:rm_to_infty:thesis_for_w}
    \end{gather}
\end{theorem}
\begin{proof}
    Fix $\delta > 0$.
    First, we prove \eqref{eq:rm_to_infty:thesis_for_w},
    namely, that $u_m$ will converge to its carrying capacity uniformly as $r_m \to \infty$.
    By \eqref{eq:cds:pde:K_form:single},
    we have that
    $u_m(x,t) = w^{(r_m,b_{mm},d_m,K_m)}(x,t,u_{0,m})$.
    The value $\bar{K} = \sum_{i=1}^{m-1}b_{mi}M_i$,
    where $M_i$ is defined as in Theorem~\ref{thm:existence},
    is an upper bound for $K_m$ independent of $r_m$.
    Then, by applying Theorem~\ref{thm:limits_for_fisher-kpp_with_K}~(iv)
    we obtain \eqref{eq:rm_to_infty:thesis_for_w}.

    We now prove \eqref{eq:rm_to_infty:thesis_for_v},
    i.e., that all other components will tend to zero as $r_m$ tends to infinity.
    Fix $i\in\{1,\dots,m-1\}$.
    By \eqref{eq:cds:pde:K_form:single},
    we have that
    $u_i(x,t) = w^{(r_i,b_{ii},d_i,K_i)}(x,t,u_{0,i})$.
    For every $(x,t) \in \clOmega \times \left[ \delta, \infty \right)$,
    the term $K_i$ satisfies
    \begin{equation*}
        K_i(x,t)
        =   \sum_{\substack{j=1\\j \ne i}}^{m} b_{ij} u_j(x,t)
        \ge b_{im} u_m(x,t)
        \ge b_{im} \inf_{(x,t) \in \clOmega \times \left[ \delta, \infty \right)} u_m(x,t)
        \eqqcolon \Kunderbar_i
        \ge 0.
    \end{equation*}
    The lower bound $\Kunderbar_i$ depends on the value of the parameter $r_m$
    and by \eqref{eq:rm_to_infty:thesis_for_w} and the strict positivity of $b_{im}$
    it holds that $\lim_{r_m\to\infty} \Kunderbar_i = \infty$.

    By the upper bound \eqref{eq:limits_for_fisher-kpp_with_K:bounds}
    in Theorem~\ref{thm:limits_for_fisher-kpp_with_K},
    for every $(x,t) \in \clOmega \times \left[ \delta, \infty \right)$
    we have that
    \begin{equation*}
        u_i(x,t) =   w^{(r_i,b_{ii},d_i,K_i)}(x,t,u_{0,i})
                 \le w^{(r_i-\Kunderbar_i,b_{ii},d_i,0)}(x,t,u_{0,i}).
    \end{equation*}
    Since $\lim_{r_m\to\infty} \left(r_i-\Kunderbar_i\right) = -\infty$,
    we can conclude by \eqref{eq:limits_for_fisherKPP:r_to_minus_infty}
    that the right hand side converges uniformly to $0$
    on $\clOmega \times \left[ \delta, \infty \right)$,
    and thus so does $u_i$,
    concluding the proof of \eqref{eq:rm_to_infty:thesis_for_v}.
\end{proof}

\begin{lemma}
    \label{lem:initial_data_can_be_taken_small}
    Let $\bm{u}$
    be the classical solution of \hyperref[eq:cds]{(P)}
    with non-negative initial data $\bm{u}_0 \in \mathcal{C}(\clOmega,\mathbb{R}^m)$.
    Then, for every $\varepsilon > 0$
    there exists a time $t_1 = t_1(\varepsilon) \ge 0$
    such that
    \begin{equation*}
        \sup_{(x,t)\in\clOmega\times\left[t_1,\infty\right)} u_i(x,t)
        <
        \frac{r_i}{b_{ii}} + \varepsilon,
    \end{equation*}
    for all $i=1,\dots,m$.
\end{lemma}
\begin{proof}
    Fix $i\in\left\{1,\dots,m\right\}$.
    By \eqref{eq:cds:pde:K_form:single} we have
    $u_i(x,t) = w^{(r_i,b_{ii},d_i,K_i)}(x,t,u_{0,i})$.
    Then, since $\Kunderbar \coloneqq 0$ is always a lower bound for $K_i$,
    by applying Theorem~\ref{thm:limits_for_fisher-kpp_with_K}~(ii) we have that
    \begin{equation*}
        \limsup_{t\to\infty} \max_{x\in\clOmega} u_i(x,t) \le \frac{r_{i}}{b_{ii}}.
    \end{equation*}
    In particular, this means that
    there exists $t_{1,i} = t_{1,i}(\varepsilon) \ge 0$
    such that
    \begin{equation*}
        \sup_{(x,t)\in\clOmega\times\left[t_{1,i},\infty\right)} u_i(x,t) < \frac{r_{i}}{b_{ii}} + \varepsilon.
    \end{equation*}
    Then, we can conclude by taking $t_1 = \max_{i=1}^m t_{1,i}$.
\end{proof}

\begin{theorem}
    \label{thm:rm_large}
    Under the same hypotheses of Theorem~\ref{thm:rm_to_infty},
    there exists $r_* > 0$,
    independent of the initial conditions $\bm{u}_0$,
    such that for every $r_m > r_*$ we have that
    \begin{enumerate}
        \item the function $u_i$
              converges to zero
              uniformly on $\clOmega$ as $t \to \infty$,
              for all $i=1,\dots,m-1$;
        \item the function $u_m$
              converges to $r_m / b_{mm}$
              uniformly on $\clOmega$ as $t \to \infty$.
    \end{enumerate}
\end{theorem}
\begin{proof}
    We define
    \begin{align}
        \bar{K}_{\delta} &\coloneqq \sum_{i=1}^{m-1} b_{mi}\left(\frac{r_{i}}{b_{ii}}+\delta\right),
        \nonumber
        \\
        \tilde{r}_{\delta} &\coloneqq \bar{K}_{\delta} + b_{mm} \max_{j=1,\dots,m-1} \frac{r_j}{b_{jm}}.
        \label{eq:rm_large:r_delta_def}
    \end{align}
    We will show that
    for all $r_m > \tilde{r}_0$
    and for all initial values $\bm{u}_0$
    the assertions (i) and (ii) hold,
    so that we can take $r_* = \tilde{r}_0$.

    Fix $r_m > r_* = \tilde{r}_0$.
    Since $\tilde{r}_{\delta}$ is continuous in $\delta$,
    we can find $\delta>0$ sufficiently small such that $r_m > \tilde{r}_{\delta}$.
    First, we show that the solution becomes bounded
    independently of the initial data.
    This is done by applying Lemma~\ref{lem:initial_data_can_be_taken_small},
    obtaining that there exists $t_1 = t_1(\delta) \ge 0$ such that
    \begin{equation}
        \sup_{(x,t)\in\clOmega\times\left[t_1,\infty\right)} u_i(x,t)
        <
        \frac{r_i}{b_{ii}} + \delta,
        \quad
        \text{for all } i=1,\dots,m.
        \label{eq:rm_large:restarted_time}
    \end{equation}

    The next step is to show that the density of the $m$-th species
    will become larger than $(\tilde{r}_{\delta}-\bar{K}_{\delta})/b_{mm}$
    for large times.
    By equation \eqref{eq:cds:pde:K_form:single},
    we have that
    $u_m(x,t) = w^{(r_m,b_{mm},d_m,K_m)}(x,t,u_{0,m})$.
    From \eqref{eq:rm_large:restarted_time}
    we get that
    the value $\bar{K}_{\delta}$ is an upper bound
    for $K_m = \sum_{i=1}^{m-1} b_{mi}u_i$.
    Moreover, since $b_{mm} > 0$ and $r_i > 0$ for every $i=1,\dots,m-1$,
    by definition of $\tilde{r}_{\delta}$
    we have $r_m > \tilde{r}_{\delta} > \bar{K}_{\delta}$.
    Then, by applying Theorem~\ref{thm:limits_for_fisher-kpp_with_K}~(iii)
    we obtain that
    \begin{equation*}
        \frac{\tilde{r}_{\delta}-\bar{K}_{\delta}}{b_{mm}}
        <
        \frac{r_m-\bar{K}_{\delta}}{b_{mm}}
        \le
        \adjustlimits
        \liminf_{t \to \infty}
            \min_{x\in\clOmega}
               u_m(x,t).
    \end{equation*}
    This means that there exists $t_2 \ge t_1$
    such that
    for all $(x,t) \in \clOmega \times \left[t_2,\infty\right)$
    it holds that
    \begin{equation}
        u_m(x,t) \ge \frac{\tilde{r}_{\delta}-\bar{K}_{\delta}}{b_{mm}}.
        \label{eq:rm_large:lower_bound_for_um}
    \end{equation}

    Now we prove (i), i.e., that the first $m-1$ species will disappear in the long run.
    This is possible since the lower bound \eqref{eq:rm_large:lower_bound_for_um}
    that we found for $u_m$ is large enough to make the total growth rates of the first $m-1$ species all negative.
    Fix $i\in\{1,\dots,m-1\}$.
    By \eqref{eq:cds:pde:K_form:single},
    we have that
    $u_i(x,t) = w^{(r_i,b_{ii},d_i,K_i)}(x,t,u_{0,i})$.
    In view of the lower bound \eqref{eq:rm_large:lower_bound_for_um}
    and of the definition \eqref{eq:rm_large:r_delta_def},
    on the set $\clOmega \times \left[ t_2, \infty \right)$
    the term $K_i$ satisfies
    \begin{equation*}
        K_i
        =   \sum_{\substack{j=1\\j \ne i}}^{m} b_{ij} u_j
        \ge b_{im} u_m
        \ge b_{im} \frac{\tilde{r}_{\delta}-\bar{K}_{\delta}}{b_{mm}}
        \ge b_{im} \max_{j=1,\dots,m-1} \frac{r_j}{b_{jm}}
        \ge r_i.
    \end{equation*}
    Then, we can choose $\Kunderbar_i \coloneqq r_i$ as a lower bound for $K_i$
    and we deduce (i)
    from Theorem~\ref{thm:limits_for_fisher-kpp_with_K}~(i).

    Finally, we prove (ii).
    Since we have shown that the first $m-1$ will go extinct for large times,
    we can now show that the last species is free to reach its carrying capacity.
    Fix $\varepsilon > 0$ arbitrarily small.
    By (i)
    there exists $t_3 \ge t_2$
    such that
    $u_i(x,t) \le \varepsilon$
    for all $(x,t) \in \clOmega \times \left[t_3,\infty\right)$
    and all $i=1,\dots,m-1$.
    Then, for every $(x,t) \in \clOmega \times \left[t_3,\infty\right)$, we have
    \begin{equation*}
        0 \le       K_m(x,t)
           =        \sum_{i=1}^{m-1} b_{mi} u_i(x,t)
          \le       \varepsilon \sum_{i=1}^{m-1} b_{mi}
          \eqqcolon \bar{K}^{(\varepsilon)}.
    \end{equation*}
    Since
    $r_m > \bar{K}^{(\varepsilon)}$
    for $\varepsilon$ small enough,
    we can apply
    Theorem~\ref{thm:limits_for_fisher-kpp_with_K}~(iii)
    to obtain
    \begin{equation*}
        \frac{r_m-\bar{K}^{(\varepsilon)}}{b_{mm}}
        \le
        \adjustlimits
        \liminf_{t \to \infty}
            \min_{x\in\clOmega}
              u_m(x,t)
        \le
        \adjustlimits
        \limsup_{t \to \infty}
            \max_{x\in\clOmega}
               u_m(x,t)
        \le
        \frac{r_m}{b_{mm}}.
    \end{equation*}
    Since $\bar{K}^{(\varepsilon)}$ tends to $0$
    as $\varepsilon \to 0$ and $\varepsilon$ is arbitrary,
    we conclude (ii).
\end{proof}

\section{Behaviour when $r_m$ is small}
\label{sec:rm_small}
In this section
we study the behaviour of \hyperref[eq:cds]{(P)}
in the case $r_m$ is very small.
Our ecological intuition tells us that
the last species will disappear
and the behaviour of the other species will be described by a $(m-1)$-species competition-diffusion system
(which may, or may not, allow for coexistence).
As in Section~\ref{sec:rm_large},
we first examine a singular limit problem in which time is bounded
and the parameter $r_m$ tends to zero.
Then, we will consider the large-time behaviour of \hyperref[eq:cds]{(P)} in the case where $r_m$ is small but positive.
In the latter case we will be able to prove our ecologically-motivated prediction
only in the case that all diffusion coefficients are equal.
This additional assumption is due to mathematical reasons
and we expect that the same result also holds in the case of unequal diffusivities.

\begin{theorem}
    Let $u_{0,m}=u_{0,m}(r_m)\in \mathcal{C}(\clOmega)$ be such that
    \begin{equation}
        0 \le u_{0,m} \le \frac{r_m}{b_{mm}}
        \quad
        \text{in } \clOmega.
        \label{eq:varying_initial_data}
    \end{equation}
    Let $\bm{u}=(\bm{v},u_m)$
    be the classical solution of \hyperref[eq:cds]{(P)},
    with initial data $\bm{u}_0=(\bm{v}_0,u_{0,m})$
    such that $\bm{v}_0\in \mathcal{C}(\clOmega,\mathbb{R}^{m-1})$ is non-negative and independent of $r_m$.
    Then, for every $\delta,T$ such that $0 < \delta < T$,
    the solution $\bm{u}$
    converges to $(\widehat{\bm{v}},0)$
    uniformly on $\clOmega \times \left[\delta,T\right]$
    as $r_m \to 0$,
    where $\widehat{\bm{v}}$
    is the classical solution of the initial value problem
    for the $(m-1)$-species competition-diffusion system
    \begin{equation}
          \partial_t \widehat{v}_i
          =
          d_i \, \upDelta \widehat{v}_i + \left( r_i - \sum_{j=1}^{m-1} b_{ij} \widehat{v}_i \right) \, \widehat{v}_i
          \quad \text{in } \upOmega\times\left(0,\infty\right),
          \quad \text{for all } i=1,\dots,m-1,
          \label{eq:rm_to_zero:m-1_system}
    \end{equation}
    with zero-flux boundary conditions on $\partial\upOmega$
    and initial conditions
    $\widehat{\bm{v}}(x,0)=\bm{v}_{0}(x)$
    for every $x\in\upOmega$.
    Moreover, the following error estimates hold
    \begin{gather}
        {\left\lVert \bm{v}-\widehat{\bm{v}} \right\rVert}_{L^{\infty}(0, T; L^2(\upOmega,\mathbb{R}^{m-1}))} \le C \, r_m,
        \label{eq:rm_small:error_estimate:v}
        \\
        {\left\lVert u_m                       \right\rVert}_{\mathcal{C}(\clOmega \times \left[ 0, T \right])} \le \frac{r_m}{b_{mm}},
        \label{eq:rm_small:error_estimate:um}
    \end{gather}
    where $C > 0$ is independent of $r_m$ and $\delta$.
\end{theorem}
\begin{proof}
    We first observe that \eqref{eq:rm_small:error_estimate:um}
    immediately follows from \eqref{eq:boundedness} and \eqref{eq:varying_initial_data},
    since they imply that for every $r_m > 0$ we have
    \begin{equation}
        0 \leq u_m \leq \frac{r_m}{b_{mm}} \quad \text{in } \upOmega \times \left[ 0, T \right].
        \label{eq:rm_to_zero:convegence_of_um_bounds}
    \end{equation}
    In particular,
    this means that
    $u_m \to 0$
    uniformly on $\clOmega \times \left[\delta,T\right]$
    as $r_m \to 0$.


    Now we will prove convergence of the first $m-1$ components.
    In this proof we will denote by $\left\langle\cdot,\cdot\right\rangle$ and $\left\lvert\cdot\right\rvert$
    the scalar product and the norm of $L^2(\upOmega,\mathbb{R}^{m-1})$ respectively.
    For simplifying the computation, we will use the following vector
    notation for the system \eqref{eq:rm_to_zero:m-1_system}:
    \begin{equation}
        \widehat{\bm{v}}_t
        =
          \bm{D} \, \upDelta \widehat{\bm{v}}
        + \bm{f}(\widehat{\bm{v}},0)
        \quad\text{in }\upOmega\times\left(0,T\right),
        \label{eq:rm_to_zero:m-1_system:vector}
    \end{equation}
    where $\bm{D} = \diag\left\{d_1,\dots,d_{m-1}\right\}$
    is the diagonal diffusion matrix
    and $\bm{f}:\mathbb{R}^m \to \mathbb{R}^{m-1}$
    is the reaction term
    defined as
    \begin{equation*}
        f_i(\widehat{\bm{v}},u_m)
        =
        \left( r_i - \sum_{j=1}^{m-1} b_{ij} \widehat{v}_i - b_{im} u_m \right) \widehat{v}_i,
        \quad
        i=1,\dots,m-1.
    \end{equation*}
    It can be easily seen that $\bm{v}$ satisfies instead
    \begin{equation}
        \bm{v}_t
        =
          \bm{D} \, \upDelta \bm{v}
        + \bm{f}(\bm{v},u_m)
        \quad\text{in }\upOmega\times\left(0,T\right).
        \label{eq:rm_to_zero:m-1_subsystem:vector}
    \end{equation}

    Let $\bm{\varphi} = \bm{v} - \widehat{\bm{v}}$.
    By subtracting \eqref{eq:rm_to_zero:m-1_subsystem:vector} and \eqref{eq:rm_to_zero:m-1_system:vector}
    we get
    \begin{equation*}
        \bm{\varphi}_t = \bm{D} \, \upDelta \bm{\varphi} + \bm{f}(\bm{v},u_m) - \bm{f}(\widehat{\bm{v}},0).
    \end{equation*}
    By taking the scalar product by $\bm{\varphi}$ in $L^2(\upOmega,\mathbb{R}^{m-1})$,
    i.e., by multiplication by $\bm{\varphi}$ and integration on $\upOmega$,
    we obtain
    \begin{equation}
        \frac{1}{2} \frac{\mathrm{d}}{\mathrm{d}t}\left\lvert\bm{\varphi}\right\rvert^{2}
        =
        \left\langle \bm{\varphi}_t, \bm{\varphi} \right\rangle
        =
        \left\langle \bm{D} \, \upDelta \bm{\varphi}, \bm{\varphi} \right\rangle
        + \left\langle \bm{f}(\bm{v},u_m) - \bm{f}(\widehat{\bm{v}},0) , \bm{\varphi} \right\rangle.
        \label{eq:rm_to_zero:norm_phi_equation}
    \end{equation}
    By Green's first identity we have
    \begin{equation}
        - \left\langle \bm{D} \, \upDelta \bm{\varphi}, \bm{\varphi} \right\rangle
        =
        \left\langle \bm{D} \, \nabla \bm{\varphi}, \nabla \bm{\varphi} \right\rangle
        \ge
        d_{\min} \left\langle \nabla \bm{\varphi}, \nabla \bm{\varphi} \right\rangle
        =
        d_{\min} \, {\left\lvert \nabla \bm{\varphi} \right\rvert}^2
        \ge
        0,
        \label{eq:rm_to_zero:norm_phi_equation:green}
    \end{equation}
    where $d_{\min} = \min_{i=1}^{m-1} d_i$,
    and thus
    \begin{equation*}
        \frac{\mathrm{d}}{\mathrm{d}t} \, {\left\lvert \bm{\varphi}\right\rvert}^2
        \le
        2 \, \lvert \bm{f}(\bm{v},u_m) - \bm{f}(\widehat{\bm{v}},0) \rvert \, \lvert\bm{\varphi}\rvert
        \le
        {\lvert \bm{f}(\bm{v},u_m) - \bm{f}(\widehat{\bm{v}},0) \rvert}^2 +  {\lvert\bm{\varphi}\rvert}^2,
    \end{equation*}
    where in the first inequality we have applied \eqref{eq:rm_to_zero:norm_phi_equation:green}
    and the Cauchy-Schwarz inequality to \eqref{eq:rm_to_zero:norm_phi_equation}.

    In the rest of the proof
    we will suppose that $r_m < 1$;
    since we are considering the limit for $r_m \to 0$,
    this can be done without any loss of generality.
    Then, since $\bm{v}_0$ is independent of $r_m$,
    it follows from \eqref{eq:boundedness}
    that $(\bm{v},u_m)$ assumes values
    in a bounded set $\mathscr{D}$
    independent of $r_m$.
    Since $\bm{f}$ is smooth,
    its restriction to $\mathscr{D}$ is Lipschitz continuous.
    Then, there exists $C'>0$ independent of $r_m$ such that
    \begin{align*}
        \frac{\mathrm{d}}{\mathrm{d}t} {\left\lvert \bm{\varphi}\right\rvert}^2
        &\le
        C' {\left\lvert (\bm{v},u_m) - (\widehat{\bm{v}},0) \right\rvert}^2 + {\left\lvert\bm{\varphi}\right\rvert}^2
        \le
        C' {\left\lvert (\bm{\varphi},u_m) \right\rvert}^2 + {\left\lvert\bm{\varphi}\right\rvert}^2
        \\ &=
        C' \left(
        {\left\lvert \bm{\varphi} \right\rvert}^2
        + {\left\lvert u_m \right\rvert}^2
        \right) + {\left\lvert\bm{\varphi}\right\rvert}^2
        \le
        (C'+1) \, {\left\lvert \bm{\varphi} \right\rvert}^2
        + C' {\left(\frac{r_m}{b_{mm}}\right)}^2 \left\lvert\upOmega\right\rvert,
    \end{align*}
    where \eqref{eq:rm_to_zero:convegence_of_um_bounds} has been substituted
    in the last inequality.
    By applying the Gronwall lemma and using the fact that ${\left\lvert \bm{\varphi}(0) \right\rvert} = 0$ since the initial data is the same for both $\bm{v}$ and $\widehat{\bm{v}}$, we obtain
    \begin{equation*}
        {\left\lvert \bm{\varphi}(t) \right\rvert}^2
        \le
        \frac{C'}{C'+1} {\left(\frac{r_m}{b_{mm}}\right)}^2 \left\lvert\upOmega\right\rvert \left( e^{(C'+1)  t} - 1\right).
    \end{equation*}
    Then,
    for $C$ such that $C^{2} = C' \, b_{mm}^{-2} \left\lvert\upOmega\right\rvert ( e^{(C'+1)T} - 1 ) / (C'+1)$
    we have that ${\left\lvert \bm{\varphi}(t) \right\rvert} \le C \, r_m$ for all $t \in \left(0,T\right)$,
    from which the error estimate \eqref{eq:rm_small:error_estimate:v} follows.
    As a consequence, we have that $\bm{v} \to \widehat{\bm{v}}$ in $L^\infty(0,T;L^2(\upOmega,\mathbb{R}^{m-1}))$ as $r_m \to 0$.

    Moreover, it follows from \cite[Theorem~5.1.17]{lunardi}
    that for every $\alpha \in\left(0,1\right)$
    there exists a positive constant $C''$ independent of $r_m$
    such that
    \begin{equation*}
        {\left\lVert \bm{v} \right\lVert}_{\mathcal{C}^{\alpha,\frac{\alpha}{2}}(\clOmega \times [\delta,T])} \le C''.
    \end{equation*}
    Thus $\left\{\bm{v}\right\}$ is bounded and equicontinuous in $\mathcal{C}(\clOmega\times[\delta,T])$.
    Applying the Ascoli-Arzelà theorem,
    we deduce that $\bm{v}$ converges uniformly to $\widehat{\bm{v}}$ on $\clOmega\times\left[\delta,T\right]$
    as $r_m \to 0$.
\end{proof}

We expect that
the $m$-th species will become extinct in the long run
even if $r_m$ is positive but sufficiently small.
We give a proof of this fact
under the additional assumption
that the total density $\sum_{i=1}^m u_i$
becomes bounded from below by a strictly positive value
for large enough times.
Such a property roughly means that,
at any given space position,
at least one species will survive in the long run,
which seems to be a reasonable outcome
for the mathematical model of species competition \eqref{eq:cds:pde}.
Such lower bounds have been recently proved
for the one-dimensional traveling wave equation associated to \eqref{eq:cds:pde}
by using an N-barrier comparison principle \wrapcite{\cite{nbarrier}},
but, as far as we are aware of, the proof for the parabolic case is still an open problem.
However, if all diffusion coefficients are equal,
we can apply the comparison principle
and easily obtain a lower bound.
\begin{lemma}
    \label{lem:non_extinction}
    Let $\bm{u}$ be a classical solution of \hyperref[eq:cds]{(P)}
    with a non-negative initial function $\bm{u}_0$ not identically equal to zero,
    equal diffusion coefficients $d_1 = \dots = d_m = d$
    and intrinsic growth rates $r_i$ depending on space and time.
    Moreover, suppose that the $r_i$'s are bounded away from zero,
    that is,
    for each $i=1,\dots,m$
    there exists $r_{i,\min} > 0$
    such that $r_i \ge r_{i,\min}$ in $\clOmega \times \left[0,\infty\right)$.

    Let $\bm{\beta} \in \mathbb{R}^m$
    be such that $\beta_i > 0$ for all $i=1,\dots,m$
    and let $\varphi=\sum_{i=1}^m \beta_i u_i$.
    Then, there exists $\xi > 0$
    such that
    for every $\delta \in \left(0,\xi\right)$
    there exists $t_2 = t_2(\delta) > 0$ such that
    $\varphi > \xi - \delta$
    in $\clOmega \times \left[ t_2, \infty \right)$.
    Moreover,
    we can choose $\xi$ independent of the initial condition $\bm{u}_0$
    and depending on the growth rates $r_i$
    only through their lower bounds $r_{i,\min}$.
    A possible choice of the constant $\xi$ is given by
    \begin{equation*}
        \xi = \left( \min_{i=1,\dots,m} r_{i,\min} \right)
              \left( \min_{i,j=1,\dots,m} \frac{\beta_j}{b_{ij}} \right).
    \end{equation*}
\end{lemma}
\begin{remark}
    We cannot let any $\beta_i$ be equal to zero,
    because a priori we do not know which subset of the species
    will not go extinct.
\end{remark}
\begin{remark}
    We will not discuss the existence and uniqueness of classical solutions of problem \hyperref[eq:cds]{(P)}
    in the case the growth rates are allowed to depend on space and/or time.
    As was the case for problem \hyperref[eq:fkpp_with_K]{(F-KPP-K)},
    this is not really needed since we will always apply Lemma~\ref{lem:non_extinction}
    to a subset of a solution of problem \hyperref[eq:cds]{(P)} with constant growth rates,
    whose existence and uniqueness was already established in Theorem~\ref{thm:existence}.
\end{remark}
\begin{proof}
    By multiplying the equation for $u_i$ in \eqref{eq:cds:pde} by $\beta_i$
    for all $i=1,\dots,m$ and then summing the resulting equations we obtain
    \begin{equation}
        \varphi_t
        =
        d \, \upDelta \varphi
        +
        \sum_{i=1}^m \beta_i \left( r_i - \sum_{j=1}^m b_{ij}u_j \right) u_i.
        \label{lem:non_extinction:equation_for_sum}
    \end{equation}
    Let $\phiunderbar$ be the solution
    to the initial problem for the Fisher-KPP reaction-diffusion equation
    \begin{equation}
        {\phiunderbar}_t = d \, \upDelta \phiunderbar + \rho \left( 1 - \frac{1}{\xi} \, \phiunderbar \right) \phiunderbar,
        \label{lem:non_extinction:lower_bound_equation}
    \end{equation}
    where $\rho,\xi > 0$,
    with zero-flux boundary conditions on $\partial\upOmega$
    and initial conditions
    $\phiunderbar(x,0) = \varphi(x,0)$
    for all $x \in \clOmega$.
    Since $\bm{u}_0$ is not identically equal to zero
    and $\bm{\beta} > \bm{0}$ componentwise,
    we have that $\varphi(x,0)$ is not identically equal to zero in $\upOmega$.

    We want to find suitable values for $\rho$ and $\xi$
    so that $\varphi$ is an upper solution of the initial value problem associated with \eqref{lem:non_extinction:lower_bound_equation}.
    Since the initial and boundary conditions coincide,
    we just need to find $\rho$ and $\xi$ such that
    \begin{equation*}
        \varphi_t \ge d \, \upDelta \varphi + \rho \left( 1 - \frac{1}{\xi} \, \varphi \right) \varphi.
    \end{equation*}
    By substituting \eqref{lem:non_extinction:equation_for_sum}, we see that
    this is true if and only if
    \begin{equation*}
        \rho \left( 1 - \frac{1}{\xi} \, \varphi \right) \varphi
        \le
        \sum_{i=1}^m \beta_i \left( r_i - \sum_{j=1}^m b_{ij} u_j \right) u_i.
    \end{equation*}
    By expanding $\varphi$ and performing the multiplications, this becomes
    \begin{equation*}
        \rho \, \sum_{i=1}^m     \beta_i u_i - \frac{\rho}{\xi} \, \sum_{i,j=1}^m \beta_i \beta_j u_i u_j
        \le
                \sum_{i=1}^m r_i \beta_i u_i -                     \sum_{i,j=1}^m b_{ij}  \beta_i u_i u_j.
    \end{equation*}
    Then, in order for $\varphi$ to be an upper solution,
    it suffices to have each of the sums on the left-hand side
    smaller than the corresponding one on the right-hand side.
    For the first sum, this is true if
    $\rho \le r_i$ in $\upOmega \times \left( 0, \infty \right)$ for every $i=1,\dots,m$.
    Since the $r_i$ are bounded from below, this can be reduced to
    $\rho \le r_{i,\min}$ for every $i=1,\dots,m$,
    which is satisfied by
    $\rho = \min_{i=1}^m r_{i,\min}$.
    In a similar way, the second sum yields
    $\xi = \rho \min_{i,j=1}^m \beta_j / b_{ij}$,
    where the minimum is taken only on the indices such that $b_{ij} > 0$.
    Note that since $r_{i,\min}, \beta_i > 0$ for all $i=1,\dots,m$,
    we also have $\rho,\xi > 0$.

    Then, we obtain that $\varphi \ge \phiunderbar$ by the comparison principle.
    Since
    $\phiunderbar = w^{(\rho,\rho/\xi,d,0)}$
    we can apply Theorem~\ref{thm:limits_for_fisherKPP}~(i)
    obtaining that
    $\phiunderbar$ converges to $\xi$
    uniformly on $\clOmega$ as $t\to\infty$.
    Since $\phiunderbar$ is a lower bound for $\varphi$,
    we conclude that for all $\delta\in\left(0,\xi\right)$
    there exists $t_2 = t_2(\delta) > 0$ such that $\varphi > \xi - \delta$
    in $\clOmega \times \left[t_2,\infty\right)$.
\end{proof}

We conclude by showing that,
if the first $m-1$ diffusion coefficients are equal,
the $m$-th species will go extinct in the long run
whenever its intrinsic growth rate $r_m$ is sufficiently small.
We remark that in the proof we do not actually use the equal diffusion hypothesis,
except in order to apply Lemma~\ref{lem:non_extinction}.
Proving Lemma~\ref{lem:non_extinction} in the case of unequal diffusion coefficients
(possibly with a different choice of $\xi$)
is, to the best of our knowledge, still an open problem,
but it would allow us to immediately extend the following theorem to the general case.

\begin{theorem}
    Let $\bm{u}=(\bm{v},u_m)$
    be the classical solution of \hyperref[eq:cds]{(P)}
    with non-negative initial data $\bm{u}_0=(\bm{v}_0,u_{0,m})\in \mathcal{C}(\clOmega,\mathbb{R}^m)$
    such that $\bm{v}_0$ is not identically equal to $\bm{0}$.
    Suppose that
    $d_1 = \dots = d_{m-1} = d$
    and that $b_{mi} > 0$ for all $i=1,\dots,m-1$.

    Then, there exists $r_* > 0$ independent of the initial data
    such that for every $0 \le r_m < r_*$
    the function $u_m$
    converges to zero
    uniformly on $\clOmega$
    as $t \to \infty$.
\end{theorem}
\begin{proof}
    We will first show by applying Lemma~\ref{lem:non_extinction}
    that if $r_m$ is sufficiently small
    then the first $m-1$ species cannot all go extinct.
    In particular, their total density will be bounded from below,
    independently of the value of $r_m$.
    The function $\bm{u}$ satisfies the competition-diffusion system \eqref{eq:cds:pde}, that is
    \begin{equation*}
        \partial_t u_i
        =
        d_i \, \upDelta u_i + \left( r_i - \sum_{j=1}^{m} b_{ij} u_j \right) u_i,
        \quad
        \text{for all } i=1,\dots,m.
    \end{equation*}
    Let $\widetilde{r}_i = r_i - b_{im} u_m$ for $i=1,\dots,m-1$.
    Then, we rewrite the equations for the first $m-1$ species as
    \begin{equation}
        \partial_t u_i
        =
        d \, \upDelta u_i + \left( \widetilde{r}_i - \sum_{j=1}^{m-1} b_{ij} u_j \right) u_i,
        \quad
        \text{for all } i=1,\dots,m-1,
        \label{eq:rm_small:m_minus_1_species_system}
    \end{equation}
    where the coupling with the $m$-th equations
    is now contained in the non-constant growth rates $\tilde{r}_i$'s.
    In order to apply Lemma~\ref{lem:non_extinction} to \eqref{eq:rm_small:m_minus_1_species_system},
    the $\widetilde{r}_i$'s must be bounded away from zero.
    We will show that this is true,
    but only when $r_m$ is small enough and in general only for times larger than a certain $t_1 \ge 0$.

    First, we search for $r_{\sharp}, \rho, \varepsilon > 0$ satisfying
    \begin{equation*}
        r_i - b_{im} \left(\frac{r_m}{b_{mm}} + \varepsilon \right) > \rho,
        \quad
        \text{for all }  i=1,\dots,m-1,
        \text{ for all } r_m \in \left[0,r_{\sharp}\right).
    \end{equation*}
    In the case $b_{im} = 0$ for all $i=1,\dots,m-1$,
    we can choose $\rho < \min_{i=1}^{m-1} r_i$ and $r_{\sharp}$ and $\varepsilon$ arbitrarily.
    Otherwise, $b_{im} > 0$ for at least one $i=1,\dots,m-1$
    and we can take
    \begin{equation*}
        r_{\sharp} = b_{mm} \left( \min_{i=1,\dots,m-1} \frac{r_i - \rho}{b_{im}} - \varepsilon \right),
    \end{equation*}
    where the minimum is taken only on the indices $i$ such that $b_{im} > 0$.
    Since $r_i > 0$ for all $i=1,\dots,m-1$,
    if we take $\rho$ and $\varepsilon$ sufficiently small
    we have that $r_{\sharp}$ is strictly positive as required.
    Note that $r_{\sharp}$, $\rho$ and $\varepsilon$
    can be chosen independently of the initial data $\bm{u}_0$.

    Now suppose that $r_m < r_{\sharp}$.
    By Lemma~\ref{lem:initial_data_can_be_taken_small}
    there exists a time $t_1 = t_1(\varepsilon) \ge 0$
    such that
    \begin{equation*}
        u_m < \frac{r_m}{b_{mm}} + \varepsilon
        \quad
        \text{in }
        \clOmega \times \left[ t_1,\infty \right),
    \end{equation*}
    and thus, for all $i=1,\dots,m-1$, we have that
    \begin{equation*}
        \widetilde{r}_i > r_i - b_{im} \left( \frac{r_m}{b_{mm}} + \varepsilon\right) > \rho
        \quad
        \text{in }
        \clOmega \times \left[ t_1,\infty \right).
    \end{equation*}
    We have thus obtained a lower bound for $\widetilde{r}_i$, $i=1,\dots,m-1$,
    which is satisfied when $t$ is large enough
    and which is independent of the choice of $r_m \in \left[0,r_{\sharp}\right)$
    and of the initial conditions.

    Then, we apply Lemma~\ref{lem:non_extinction}
    to the $(m-1)$-species competition-diffusion system \eqref{eq:rm_small:m_minus_1_species_system}
    with $\beta_i = b_{mi} > 0$, $i=1,\dots,m-1$,
    restarting time from the time $t_1$.
    This is possible because of the lower bound for the intrinsic growth rates just obtained
    and thanks to the fact that $\bm{v}(\cdot,t_1)$ is not identically zero
    by \eqref{eq:solution_is_strictly_positive}.
    As a result, there exists
    a value $\xi > 0$
    (independent of $r_m$ and of the initial conditions $\bm{u}_0$)
    and a time $t_2$
    such that
    \begin{equation}
      \sum_{i=1}^{m-1} b_{mi} u_i > \xi \quad \text{in } \clOmega \times \left[ t_2, \infty \right).
      \label{eq:rm_small:lowerbound_for_m_minus_1_species}
    \end{equation}

    We set $r_* = \min\left\{r_{\sharp},\xi \right\}$
    or, in the case $r_{\sharp}$ could be chosen arbitrarily,
    $r_* = \xi$.
    As the last step, we need to show
    that if $r_m \in \left[ 0, r_* \right)$
    the $m$-th species disappears as $t\to\infty$.
    By \eqref{eq:cds:pde:K_form:single},
    we have that
    $u_m(x,t) = w^{(r_m,b_{mm},d_m,K_m)}(x,t,u_{0,m})$.
    By \eqref{eq:rm_small:lowerbound_for_m_minus_1_species}
    and the definition of $r_*$,
    we have that $K_m = \sum_{i=1}^{m-1} b_{mi} u_i > \xi \ge r_* > r_m$
    in $\clOmega \times \left[ t_2, \infty \right)$.
    Then, the proof is completed
    by applying Theorem~\ref{thm:limits_for_fisher-kpp_with_K}~(i)
    with the lower bound $\Kunderbar \coloneqq \xi$ for $K_m$.
\end{proof}

\ifspringer
\else
\section*{Acknowledgements}
\MyAck

\printbibliography{}
\fi

\end{document}